\documentclass[10pt]{article}

\usepackage{times}

\usepackage{stmaryrd}
\usepackage{enumerate}
\usepackage{algorithm}
\usepackage{algpseudocode}
\usepackage{blkarray}
\usepackage{amsmath}
\usepackage[titletoc,title]{appendix}
\usepackage{graphicx,epic,eepic,epsfig,latexsym,amssymb,verbatim,color}\usepackage{subfigure}
\usepackage{mathrsfs}
\usepackage{dsfont}
\usepackage{float}
\usepackage{cases}
\usepackage{makecell}
\usepackage{cases}
\usepackage{tikz}
\usepackage{qcircuit}
\usepackage{hyperref}
\usepackage{cleveref}
\hypersetup{colorlinks=true,citecolor=blue,linkcolor=blue,filecolor=blue,urlcolor=blue,breaklinks=true}

\usetikzlibrary{calc, shapes, backgrounds}
\makeatletter
\def\grd@save@target#1{%
  \def\grd@target{#1}}
\def\grd@save@start#1{%
  \def\grd@start{#1}}
\tikzset{
  grid with coordinates/.style={
    to path={%
      \pgfextra{%
        \edef\grd@@target{(\tikztotarget)}%
        \tikz@scan@one@point\grd@save@target\grd@@target\relax
        \edef\grd@@start{(\tikztostart)}%
        \tikz@scan@one@point\grd@save@start\grd@@start\relax
        \draw[minor help lines,magenta] (\tikztostart) grid (\tikztotarget);
        \draw[major help lines] (\tikztostart) grid (\tikztotarget);
        \grd@start
        \pgfmathsetmacro{\grd@xa}{\the\pgf@x/1cm}
        \pgfmathsetmacro{\grd@ya}{\the\pgf@y/1cm}
        \grd@target
        \pgfmathsetmacro{\grd@xb}{\the\pgf@x/1cm}
        \pgfmathsetmacro{\grd@yb}{\the\pgf@y/1cm}
        \pgfmathsetmacro{\grd@xc}{\grd@xa + \pgfkeysvalueof{/tikz/grid with coordinates/major step}}
        \pgfmathsetmacro{\grd@yc}{\grd@ya + \pgfkeysvalueof{/tikz/grid with coordinates/major step}}
        \foreach \x in {\grd@xa,\grd@xc,...,\grd@xb}
        \node[anchor=north] at (\x,\grd@ya) {\pgfmathprintnumber{\x}};
        \foreach \y in {\grd@ya,\grd@yc,...,\grd@yb}
        \node[anchor=east] at (\grd@xa,\y) {\pgfmathprintnumber{\y}};
      }
    }
  },
  minor help lines/.style={
    help lines,
    step=\pgfkeysvalueof{/tikz/grid with coordinates/minor step}
  },
  major help lines/.style={
    help lines,
    line width=\pgfkeysvalueof{/tikz/grid with coordinates/major line width},
    step=\pgfkeysvalueof{/tikz/grid with coordinates/major step}
  },
  grid with coordinates/.cd,
  minor step/.initial=.2,
  major step/.initial=1,
  major line width/.initial=2pt,
}
\makeatother

\usepackage[marginal]{footmisc}
\usepackage{url}
\usepackage{theorem}
\newtheorem{definition}{Definition}

\newtheorem{lemma}[definition]{Lemma}

\newtheorem{theorem}[definition]{Theorem}

\def\squareforqed{\hbox{\rlap{$\sqcap$}$\sqcup$}}
\def\qed{\ifmmode\squareforqed\else{\unskip\nobreak\hfil
\penalty50\hskip1em\null\nobreak\hfil\squareforqed
\parfillskip=0pt\finalhyphendemerits=0\endgraf}\fi}
\def\endenv{\ifmmode\;\else{\unskip\nobreak\hfil
\penalty50\hskip1em\null\nobreak\hfil\;
\parfillskip=0pt\finalhyphendemerits=0\endgraf}\fi}

\newenvironment{proof}{\noindent \textbf{{Proof~}}}{\hfill $\blacksquare$}

\mathchardef\ordinarycolon\mathcode`\:
\mathcode`\:=\string"8000
\def\vcentcolon{\mathrel{\mathop\ordinarycolon}}
\begingroup \catcode`\:=\active
  \lowercase{\endgroup
  \let :\vcentcolon
  }

\makeatletter
\def\resetMathstrut@{%
	\setbox\z@\hbox{%
		\mathchardef\@tempa\mathcode`\[\relax
		\def\@tempb##1"##2##3{\the\textfont"##3\char"}%
		\expandafter\@tempb\meaning\@tempa \relax
	}%
	\ht\Mathstrutbox@\ht\z@ \dp\Mathstrutbox@\dp\z@}
\makeatother
\begingroup
\catcode`(\active \xdef({\left\string(}
\catcode`)\active \xdef){\right\string)}
\endgroup
\mathcode`(="8000 \mathcode`)="8000

\newcommand{\nc}{\newcommand}
\nc{\rnc}{\renewcommand}
\nc{\lbar}[1]{\overline{#1}}
\nc{\bra}[1]{\langle#1|}
\nc{\ket}[1]{|#1\rangle}
\nc{\ketbra}[2]{|#1\rangle\!\langle#2|}
\nc{\braket}[2]{\langle#1|#2\rangle}

\nc{\proj}[1]{| #1\rangle\!\langle #1 |}
\nc{\avg}[1]{\langle#1\rangle}
\nc{\Rank}{\operatorname{Rank}}
\nc{\smfrac}[2]{\mbox{$\frac{#1}{#2}$}}
\nc{\tr}{\operatorname{Tr}}
\nc{\ox}{\otimes}
\nc{\dg}{\dagger}
\nc{\dn}{\downarrow}
\nc{\cA}{\boldsymbol{\cal A}}
\nc{\ca}{\boldsymbol{a}}
\nc{\cB}{{\cal B}}
\nc{\cC}{{\cal C}}
\nc{\cD}{{\cal D}}
\nc{\cE}{{\cal E}}
\nc{\cF}{{\cal F}}
\nc{\cG}{{\cal G}}
\nc{\cH}{{\cal H}}
\nc{\cI}{{\cal I}}
\nc{\cJ}{{\cal J}}
\nc{\cK}{{\cal K}}
\nc{\cL}{{\cal L}}
\nc{\cM}{\boldsymbol{\cal M}}
\nc{\cN}{{\cal N}}
\nc{\cO}{{\cal O}}
\nc{\cP}{{\cal P}}
\nc{\cQ}{\boldsymbol{\cal Q}}
\nc{\cq}{\boldsymbol{q}}
\nc{\cR}{{\cal R}}
\nc{\cS}{{\cal S}}
\nc{\cT}{{\cal T}}
\nc{\cV}{{\cal V}}
\nc{\cX}{{\cal X}}
\nc{\cx}{\boldsymbol{x}}
\nc{\cY}{{\cal Y}}
\nc{\cZ}{{\cal Z}}
\nc{\cW}{{\cal W}}
\nc{\csupp}{{\operatorname{csupp}}}
\nc{\qsupp}{{\operatorname{qsupp}}}
\nc{\var}{{\operatorname{var}}}
\nc{\rar}{\rightarrow}
\nc{\lrar}{\longrightarrow}
\nc{\polylog}{{\operatorname{polylog}}}
\nc{\wt}{{\operatorname{wt}}}
\nc{\av}[1]{{\left\langle {#1} \right\rangle}}
\nc{\supp}{{\operatorname{supp}}}

\def\x{\xi}

\nc{\RR}{{{\mathbb R}}}
\nc{\CC}{{{\mathbb C}}}
\nc{\FF}{{{\mathbb F}}}
\nc{\NN}{{{\mathbb N}}}
\nc{\ZZ}{{{\mathbb Z}}}
\nc{\PP}{{{\mathbb P}}}
\nc{\QQ}{{{\mathbb Q}}}
\nc{\UU}{{{\mathbb U}}}
\nc{\EE}{{{\mathbb E}}}
\nc{\id}{{\operatorname{id}}}

\nc{\<}{\langle}
\rnc{\>}{\rangle}
\nc{\Hom}[2]{\mbox{Hom}(\CC^{#1},\CC^{#2})}
\nc{\rU}{\mbox{U}}

\nc{\ob}[1]{#1}

\nc{\SEP}{{\text{SEP}}}
\nc{\NS}{{\text{NS}}}
\nc{\LOCC}{{\text{LOCC}}}
\nc{\PPT}{{\text{PPT}}}
\nc{\EXT}{{\text{EXT}}}
\nc{\Sym}{{\operatorname{Sym}}}

\definecolor{darkblue}{RGB}{0,76,156}
\definecolor{darkkblue}{RGB}{0,0,153}
\definecolor{blue2}{RGB}{102,178,255}
\definecolor{myred}{RGB}{180,5,4}

\RequirePackage[framemethod=default]{mdframed}
\newmdenv[skipabove=7pt,
skipbelow=7pt,
innerleftmargin=5pt,
innerrightmargin=5pt,
innertopmargin=5pt,
leftmargin=0cm,
rightmargin=0cm,
innerbottommargin=5pt,
linewidth=1pt]{tBox}

\newmdenv[skipabove=7pt,
skipbelow=7pt,
backgroundcolor=blue2!25,
innerleftmargin=5pt,
innerrightmargin=5pt,
innertopmargin=5pt,
leftmargin=0cm,
rightmargin=0cm,
innerbottommargin=5pt,
linewidth=1pt]{dBox}

\newmdenv[skipabove=7pt,
skipbelow=7pt,
backgroundcolor=darkkblue!15,
innerleftmargin=5pt,
innerrightmargin=5pt,
innertopmargin=5pt,
leftmargin=0cm,
rightmargin=0cm,
innerbottommargin=5pt,
linewidth=1pt]{sBox}

\usepackage{enumitem}

\usepackage{authblk}
\newtheorem{innercustomgeneric}{\customgenericname}
\providecommand{\customgenericname}{}
\newcommand{\newcustomtheorem}[2]{%
  \newenvironment{#1}[1]
  {%
   \renewcommand\customgenericname{#2}%
   \renewcommand\theinnercustomgeneric{##1}%
   \innercustomgeneric
  }
  {\endinnercustomgeneric}
}

\newcustomtheorem{customthm}{Theorem}
\newcustomtheorem{customlemma}{Lemma}

\begin{document}
\title{Quantum Search with Prior Knowledge}

\author[1,2]{Xiaoyu He}
\author[1,2]{ Jialin Zhang}
\author[1,2]{ Xiaoming Sun\thanks{sunxiaoming@ict.ac.cn}}

\affil[1]{Institute of Computing Technology, Chinese Academy of Sciences}
\affil[2]{University of Chinese Academy of Sciences}

\date{}
\maketitle
\vspace{-1cm}
\begin{abstract}
Search-base algorithms have widespread applications in different scenarios. Grover's quantum search algorithms and its generalization, amplitude amplification, provide a quadratic speedup over classical search algorithms for unstructured search. We consider the problem of searching with prior knowledge. More preciously, search for the solution among $N$ items with a prior probability distribution. This letter proposes a new generalization of Grover's search algorithm which performs better than the standard Grover algorithm in average under this setting. We prove that our new algorithm achieves the optimal expected success probability of finding the solution if the number of queries is fixed. 
\end{abstract}
\section{Introduction}
Quantum searching is an important problem. Many computationally difficult problem, including deciphering some popular encryption scheme\cite{wiener1994efficient,daemen1999aes}, Monte Carlo tree search for game \cite{coulom2006efficient} and NP-hard problem\cite{bennett1997strengths}, can be reduced to searching problem. Grover's quantum search algorithm\cite{grover1997quantum,grover1998quantum} shows a quadratic speed-up over classical search algorithm. The problem is represented by an oracle function, which gives the solution by flipping the sign of corresponding quantum state, among an unsorted list of size $N$.

Grover's search algorithm applies oracle and preparation of initial state  iteratively to implement the rotation on the 2-dimensional Hilbert space spanned by the initial state and target state. It rotates a fixed angle, determined by size $N$ and the number of solutions, towards the target every time and it may miss the solution when the fixed angle is not exact.  Quantum amplitude amplification\cite{brassard2002quantum} generalized Grover's standard algorithm to make the success probability of search reach $100\%$. 
The exact query complexity of this algorithm to find one solution among $N$ items exactly is the same as Grover's algorithm, $\lceil\frac{\pi}{4\arcsin\sqrt{1/N}}-\frac{1}{2}\rceil\approx O(\sqrt{N})$, which is optimal\cite{zalka1999grover}.
 
These searching algorithms are also suitable for finding one of $M$ solutions among items of size $N$. However, the case is complicated when the number of solutions is unknown since too little iterations cannot reach the target state and too many iterations may pass by the target state. Grover's $\pi/3$-algorithm\cite{grover2005fixed} and a critically damped quantum search algorithm \cite{mizel2009critically} were proposed to handle this fixed point search problem. Then an amplitude amplification improvement\cite{yoder2014fixed} was made to achieve quadratic speed-up of query complexity.

Grover's quantum search algorithm and above generalizations consider all the items to be searched are of equal status. However, the items are rarely equally important. We note that heuristic methods are commonly used in classical searching algorithms. For practical search problems, some prior knowledge is usually known about where the solution is more likely to be located. Then a biased probability distribution can be obtained from the current problem to instruct how to prune branches or enter branches of the search tree randomly to improve the expected success probability of search algorithm. A famous modern example is Google's go AI: AlphaGo and AlphaZero\cite{silver2017mastering,silver2018general}, despite the emphasis of their work is machine learning, they search for the best placement instructed by knowledge drawn from the current pattern and check if a placement is good enough by some complex calculation.

This paper focuses on how to take advantage of prior knowledge to improve the success probability of quantum search. Here we formalize the prior knowledge as a probability distribution $\textbf{p}=(p_1,p_2,\dots,p_N)$, which indicates the probability about the location of the solution. Assuming $p_1\ge p_2\ge\dots \ge p_N$, then the classical maximum expected success is obviously $\sum_{i=1}^T p_i$, which is reached by querying the item from $1$ to $T$. A naive example that shows the advantage of prior knowledge is $\textbf p=(0.25,0.25,0.25,0.25,0,0,0,0)$ with $N=8$, the success probability in one oracle query with or without this prior knowledge is $100\%$ and about $78.1\%$ respectively, while the classical success probability in one query is $25\%$. 

Actually, when the distribution of the solution is uniform, repeating the iteration of Grover's algorithm can reach the maximum average success probability if the number of queries is limited to some given integer $T$. 
Our result shows that by replacing the initialization and diffusion operation of Grover's algorithm with some parametric operations related to the given probability distribution, we can improve the expected success probability to optimal for any given number of oracle queries.

Montanaro\cite{montanaro2010quantum} shows an algorithm that can find the solution with an asymptotically optimal expected number of queries. His quantum algorithm is a Las Vegas algorithm, that repeats many times until finding the solution. The expected number of oracle queries is proven a constant multiple of the optimal value. 
Nevertheless, the constant of time complexity is important in some scenarios. An algorithm with asymptotically optimal running time is not good enough for time-sensitive tasks, such as game AI which must output a solution in limited time, while running time directly determines the strength of artificial intelligence.

On the other hand, we note that on actual physical quantum devices, the fidelity of circuit decreases rapidly with the number of oracle queries increases. A three-qubit Grover search experiment on scalable quantum computing system\cite{figgatt2017complete} uses only one oracle query, it seems that search one solution among eight items using one oracle query performs better than using two oracle queries, although the relationship between their theoretical success probability is opposite. Our result can offer a trade-off between theoretical success probability and experimental error of circuit, which results in an algorithm with a suitable number of oracle queries to reach optimal experiment success probability.

This paper is organized as follows. We show the details of our search algorithm (Sec. II). When given a probability distribution of solution and a limited number of oracle queries, our algorithm reaches an optimal expected success probability. Then the proof of optimality of our method is shown in Appendix. We then show how quantum can speed up heuristic search of a game tree, while the probability distribution of solution can be approximated by some modern techniques such as machine learning(Sec. III). We also demonstrated our algorithm for some biased distribution of solution on IBM's quantum computing system IBMQ\cite{cross2018ibm}, which shows the improvement of experimental success probability on a physical device(Sec. IV).

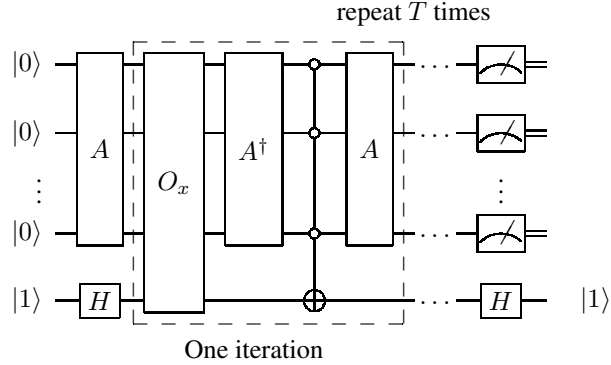
\begin{figure}[htbp] 
  \centerline{
\Qcircuit @C=0.8em @R=1.2em {
   &    &   & &              &     &   \mbox{repeat $T$ times}& & &\\
   \lstick{\ket{0}} & \multigate{3}{A}  &\multigate{4}{O_x} & \multigate{3}{A^{\dagger}}&\ctrlo{1}&\multigate{3}{A}&\qw & \dots  &  &\meter &\cw\\
   \lstick{\ket{0}} & \ghost{A}         &\ghost{O_x}	   & \ghost{A^{\dagger}}  	  &\ctrlo{2}&\ghost{A}&\qw &\dots & &\meter &\cw \\
   \lstick{\vdots}&&&&&&&&&\vdots\\
   \lstick{\ket{0}} & \ghost{A}    		&\ghost{O_x}	   & \ghost{A^{\dagger}} 	  &\ctrlo{1} 	&\ghost{A}&\qw &\dots & &\meter &\cw \\
   \lstick{\ket{1}} & \gate{H}   		&\ghost{O_x}	   & \qw 	  &\targ 	&\qw &\qw \gategroup{2}{3}{6}{6}{.8em}{--}&\dots & &\gate{H} & \qw &\rstick{\ket{1}} \\
   &                &                   &\mbox{One iteration} &                           &         & &        &  & &
}
}
  \caption{We provide a circuit for our quantum search algorithm in $T$ oracle queries. Here $A$ is an unitary that prepares the initial state $A\ket{0}=\sum_{i=1}^N \sqrt{q_i} |i\rangle + \sqrt{1-\sum_{i=1}^N q_i}|0\rangle$. $O_x$ is the oracle that flip an ancilla qubit, that $O_x\ket{x}\ket{b}=\ket{x}\ket{b\oplus1}$ and $O_x\ket{x^\perp}\ket{b}=\ket{x}\ket{b\oplus1}$}.
  \label{fig:alg}
\end{figure}
\vspace{0.cm}
\section{An optimal quantum search algorithm}
Grover's quantum search algorithm can be described as a rotation in the space spanned by initial state and the target solution state. Given the probability distribution $\textbf p=(p_1,p_2,\dots,p_N)$ as prior knowledge about where the solution might be located, we show that by replacing the initial state with $|s\rangle=\sum_{i=1}^N \sqrt{q_i} \ket{i}+ \sqrt{1-\sum_{i=1}^N q_i}\ket{0}$, where $q_i$ are some parameters that $q_i\ge 0$ and $\sum_{i=1}^N q_i\le 1$, we can improve the expected success probability. Let $R_s=I-2\ketbra{s}{s}$, which is the reflection operation about $|s\rangle$. Consider the transformation $(R_sO_x)^T$ applied to $|s\rangle$. For a fixed solution $x$, let $|x^\perp\rangle={1}/({1-q_x})\sum_{i\neq x}\sqrt{q_i}|i\rangle$, the state after $(R_sO_x)^T$ applied to $|s\rangle$ is $\sin\left((2T+1)\arcsin\sqrt{q_x}\right)|x\rangle + \cos\left((2T+1)\arcsin\sqrt{q_x}\right) |x^{\perp}\rangle$. The success probability of getting solution $x$ is $\sin^2\left((2T+1)\arcsin\sqrt{q_x}\right)$
if we use measurement with standard basis. Hence the expected success probability of our algorithm is
\begin{equation}
\sum_{x=1}^N p_i\sin^2((2T+1)\arcsin\sqrt{q_i}).
\end{equation}
Then we have following theorem.
\begin{theorem}\label{thm:lower}
Given distribution $\textbf p=(p_1,p_2,\dots,p_N)$ indicating where the solution might be, for any non-negative vector $\textbf{q}=(q_1,\dots,q_N)$ that $\sum_{i=1}^N q_i\le 1$, the expected success probability of quantum search in $T$ queries can be 
\begin{equation}\label{eq:optimization}
\text{\rm ESP}_T(\textbf p,\textbf{q})=\sum_{i=1}^N p_i\sin^2((2T+1)\arcsin\sqrt{q_i})
\end{equation}
\end{theorem}
Therefore we can choose $\textbf{q}$ to optimize the success probability. Note that $\sin^2((2T+1)\arcsin\sqrt{q_i})$ reaches its maximum value $1$ when $q_i=\sin^2\frac{\pi}{2(2T+1)}$ and it increases monotonously with $q_i$ in $[0,\sin^2\frac{\pi}{2(2T+1)}]$, so we can add extra constraints that $q_i\le \sin^2\frac{\pi}{2(2T+1)}$ , without  decreasing the maximum value of Eq.(\ref{eq:optimization}). Let $\textbf{q}^*\!(\textbf p)$ denote the parameters that maximize Eq.(\ref{eq:optimization}).

\begin{align}\label{eq:extra}
\textbf{q}^*\!(\textbf p)=\mathop{\arg\max}_{\textbf{q}}&\sum_{i=1}^np_i\sin^2((2T+1)\arcsin\sqrt{q_i})\\
\mbox{s.t. }&\sum_{i=1}^n q_i\le 1,\\
&\forall i:\ 0\le q_i\le\sin^2\frac{\pi}{2(2T+1)}.
\end{align}
Since $\sin^2((2T+1)\arcsin\sqrt{q_i})$ is concave when $q_i\in[0,\sin^2\frac{\pi}{2(2T+1)}]$, $\sum_{i=1}^np_i\sin^2((2T+1)\arcsin\sqrt{q_i})$ is a concave function about $(q_1,\dots,q_N)$. Eq~\ref{eq:extra} is a concave optimization with linear constraints, which can be solved as a convex optimization problem by Lagrange multiplier method, gradient descent or some other contemporary methods.

The perturbation of $\textbf p$ is robust to $\textbf{q}^*\!(\textbf p)$. If we have an estimation of distribution $\hat{\textbf p}$ that the $L_1$ distance $d_1(\textbf p,\hat{\textbf p})=\sum_{i=1}^N |p_i-\hat{p}_i|\le\epsilon$, the expected success probability using $\textbf{q}^*\!(\hat{\textbf p})$ is 
\begin{align*}
&\text{ESP}_T(\textbf p,\textbf{q}^*\!(\hat{\textbf p}))\\
=&\text{ESP}_T(\hat{\textbf p},\textbf{q}^*\!(\hat{\textbf p}))+\text{ESP}_T(\textbf p-\hat{\textbf p},\textbf{q}^*\!(\hat{\textbf p}))\\
\ge&\text{ESP}_T(\hat{\textbf p},\textbf{q}^*\!(\textbf p))+\text{ESP}_T(\textbf p-\hat{\textbf p},\textbf{q}^*\!(\hat{\textbf p}))\\
=&\text{ESP}_T(\textbf p,\textbf{q}^*\!(\textbf p))+\text{ESP}_T(\hat{\textbf p}-\textbf p,\textbf{q}^*\!(\textbf p))+\text{ESP}_T(\textbf p-\hat{\textbf p},\textbf{q}^*\!(\hat{\textbf p}))\\
\ge&\text{ESP}_T(\hat{\textbf p},\textbf{q}^*\!(\textbf p))-\sum_{i=1}^N|p_i-\hat{p}_i|\\
\ge&\text{ESP}_T(\hat{\textbf p},\textbf{q}^*\!(\textbf p))-\epsilon,
\end{align*}
which means a good estimation of the precise distribution can approach the optimal expected success probability well.

Our quantum search algorithm is a variation of Grover's quantum search algorithm. Our algorithm starts with initial state $|s\rangle=\sum_i^N \sqrt{q_i^*}|i\rangle$ and applies $(R_s O_x)^T$ to $|s\rangle$, where $R_s=I-2|s\rangle\langle s|$. We claim that our algorithm is optimal, which means the maximum expected success probability in $T$ queries can be reached by our algorithm. The proof of the optimality is given in Appendix.

The maximum expected success probability by $T$ classical query can be achieved by quantum algorithm using only $\lceil\sqrt{T}\rceil$ oracle queries. Let $q_i=\sin^2\frac{1}{2\lceil\sqrt{T}\rceil+1}$ if $i\le\lceil\sqrt{T}\rceil$, $q_i=0$ otherwise. Then $\text{ESP}_{\lceil \sqrt{T}\rceil}(\textbf{p},\textbf{q})\ge\sum_{i=1}^T p_i$, which means quantum has at least quadratic speed-up over classical when given probability distribution of the solution.

\section{Application to game tree search}
Solving a game is an important application of search algorithm. Ambainis\cite{ambainis2017quantum} proposed a quantum algorithm that can solve $2$-player game quadratically faster than the classical deterministic algorithm. On the other hand many heuristic methods are applied to classical search algorithms such as $k$-SAT\cite{hansen2019faster}, which can quadratically speed up by \cite{dantsin2005quantum}. Combining heuristic methods and quantum search algorithms is of great interest. Different from those algorithms above that search for a solution at leaf nodes, solving a game decides which child of root node is optimal. We show a framework of solving game by calling our algorithm recursively and give some evidence that our algorithm would perform well.

We call a choice good if a player will win the game eventually after taking it, no matter how his adversary acts afterwards. In a probabilistic search tree with edge suggesting which choice is more likely to be good, such as Monte Carlo tree after trained, we can obtain some knowledge formed as a probability distribution about each node to be good.
\begin{figure}[h!]
\centering
\begin{tikzpicture}[
    scale = 1, transform shape, thick,
    every node/.style = {draw, circle, minimum size = 6mm},
    grow = down,  
    level 1/.style = {sibling distance=4cm},
    level 2/.style = {sibling distance=2cm},
    level distance = 1.2cm
  ]
  \node[shape = circle, draw, line width = 1pt,
          minimum size = 6mm, inner sep = 0mm, font = \sffamily\large,] (Start){}
   child {   node [fill=gray!30] (A) {0}
     child { node [color=white] (B) {00}}
     child { node [color=white] (C) {01}}
   }
   child {   node [fill=gray!30] (D) {1}
     child { node [color=white] (E) {10}}
     child { node [color=white] (F) {11}}
   };

  \begin{scope}[nodes = {draw = none}]
    \path (Start) -- (A) node [near start, left]  {$0.4$};
    \path (A)     -- (B) node [near start, left]  {$0.8$};
    \path (A)     -- (C) node [near start, right] {$0.2$};
    
    \path (Start) -- (D) node [near start, right] {$0.6$};
    \path (D)     -- (E) node [near start, left]  {$0.5$};
    \path (D)     -- (F) node [near start, right] {$0.5$};
    \begin{scope}[]
      \node at (B) {$\cdots$};
      \node at (C) {$\cdots$};
      \node at (E) {$\cdots$};
      \node at (F) {$\cdots$};
    \end{scope}
  \end{scope}
\end{tikzpicture}
\vspace*{-3mm}
\caption{A probabilistic game tree where the value on the edge indicating  the probability of the choice to be good enough. }
\label{fig:tree}
\end{figure}
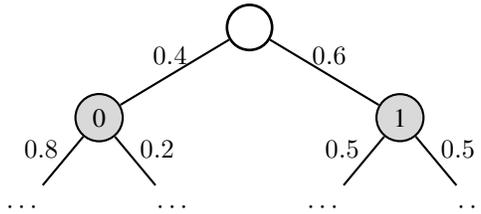

The classical search strategy at step $t$ can be formed as a probability distribution $\textbf{g}_t(s_1\dots s_t)$, where $\textbf{g}_t(s_1\dots s_t)_i$ is the probability that entering branch $i$ at pattern $s_1s_2\dots s_{t}$ for next step. These strategies can be obtained from domain expert knowledge or by machine learning. Taking cheese game as an example, when $s_1\dots s_{t}$ represents the first $t$ placement positions that two players alternately played, $\textbf{g}_t(s_1\dots s_t)_i$ is the probability that $i$ might be the best placement.  Actually, $\textbf{g}_t(s_1\dots s_{t})$ can be seen as an approximation of the real optimal solution. When the number of queries is limited to $T$, the first $T$ nodes with the highest probability can be evaluated in classical algorithm, while directly using Grover's algorithm after ranking the probabilities implies a traverse over $\Theta(T^2)$ most likely choices. Our algorithm is better than directly calling Grover's algorithm after ranking, as shown in Fig.\ref{fig:heuesp}.
\begin{figure}[htbp!]
	\centering
	\includegraphics[width=0.8\textwidth]{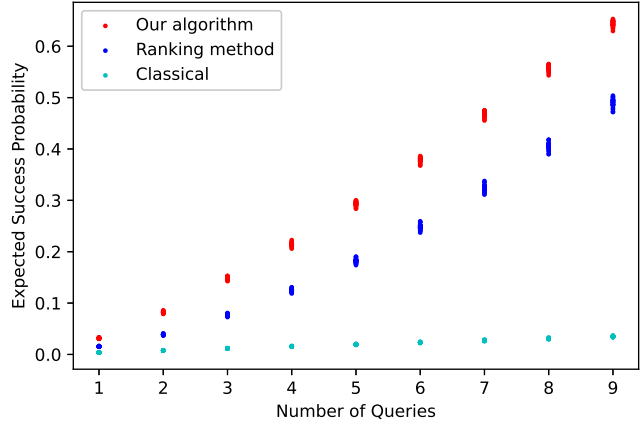}  
\caption{Comparison of the expected success probability between classical, ranking method and our algorithm for probability distributions of size 512 which are generated uniformly randomly, with $100$ samples of probability distributions for every number of queries. } 
	\label{fig:heuesp}  
\end{figure}
Since the strategy $\textbf{g}_t$ can be calculated efficiently by classical computer, the quantum operator that generate state $\sum_i\sqrt{\textbf{q}^*(\textbf{g}(s_1\dots s_t))_i}\ket{i}$ can be implemented efficiently, as shown in Fig.\ref{fig:treestate}.
\hspace{10pt}
\begin{figure}[htbp] 
  \centerline{
\Qcircuit @C=0.8em @R=1.5em {
   \lstick{\mbox{ancilla}} & {/}\qw& \multigate{4}{G_t}  &\multigate{5}{Q^*}&\multigate{4}{G_t^\dagger}  &\qw & \\
   \lstick{\ket{s_1}} &{/} \qw & \ghost{G_t}  &\ghost{Q^*}&\ghost{G_t^\dagger} &\qw &\rstick{\ket{s_1}}\\
   \lstick{\ket{s_2}} & {/}\qw & \ghost{G_t}  &\ghost{Q^*}&\ghost{G_t^\dagger} &\qw & \rstick{\ket{s_2}}\\
   \lstick{\vdots} &  &    &  &  &  \vdots & & &  & \\   
   \lstick{\ket{s_{t}}} &{/} \qw & \ghost{G_t}  &\ghost{Q^*} &\ghost{G_t^\dagger} &\qw &\rstick{\ket{s_t}} \\
   \lstick{\ket{0}} &{/} \qw & \qw  &\ghost{Q^*} &\qw &\qw &\rstick{\sum_i\sqrt{\textbf{q}^*(\textbf{g}(s_1\dots s_t))_i}\ket{i}}& &&&&&&&&&\\
}
}
  \caption{We provide a circuit for generating the optimal initial state of quantum search corresponding to distribution $\textbf{g}(s_1\dots s_t)$. $G_t$ calculate the probability distribution $\textbf{g}(s_1\dots s_t)$ as classical algorithm and store it to ancillae qubits. $Q^*$ calculate $\textbf{q}^*(\textbf{g}(s_1\dots s_t))$ and then generate the corresponding quantum state. Then ancillae qubits are recovered by reversing the circuit.}
  \label{fig:treestate}
\end{figure}
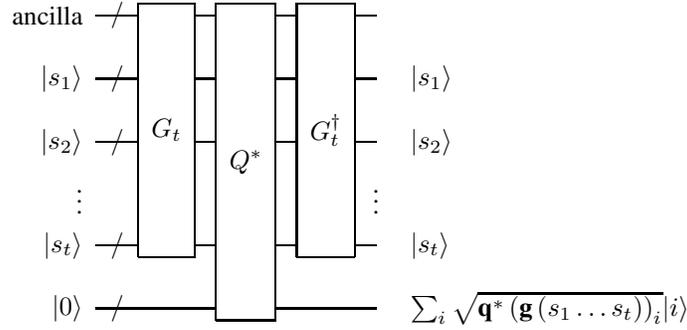

Note that when Alice plays a game with Bob, to determine whether a choice $i$ is good for Alice, Alice only needs to predict whether Bob can take a good choice after Alice takes choice $i$. This means that the oracle determines whether a choice $i$ is good for Alice can be implemented by searching for a good choice for Bob at the pattern after Alice takes choice $i$. Hence we can call our search algorithm recursively, search to a certain depth in the game tree, and then use ``approximated oracle'' such as neural network or rule-based evaluation. This implies a quantum algorithm framework similar to practical classical game tree search, which can speed up the game tree search. Our algorithm is optimal at each level of the tree locally, we believe it would perform well globally as well.

\begin{figure*}[htbp!]
\subfigure[]{
	\includegraphics[width=0.47\textwidth]{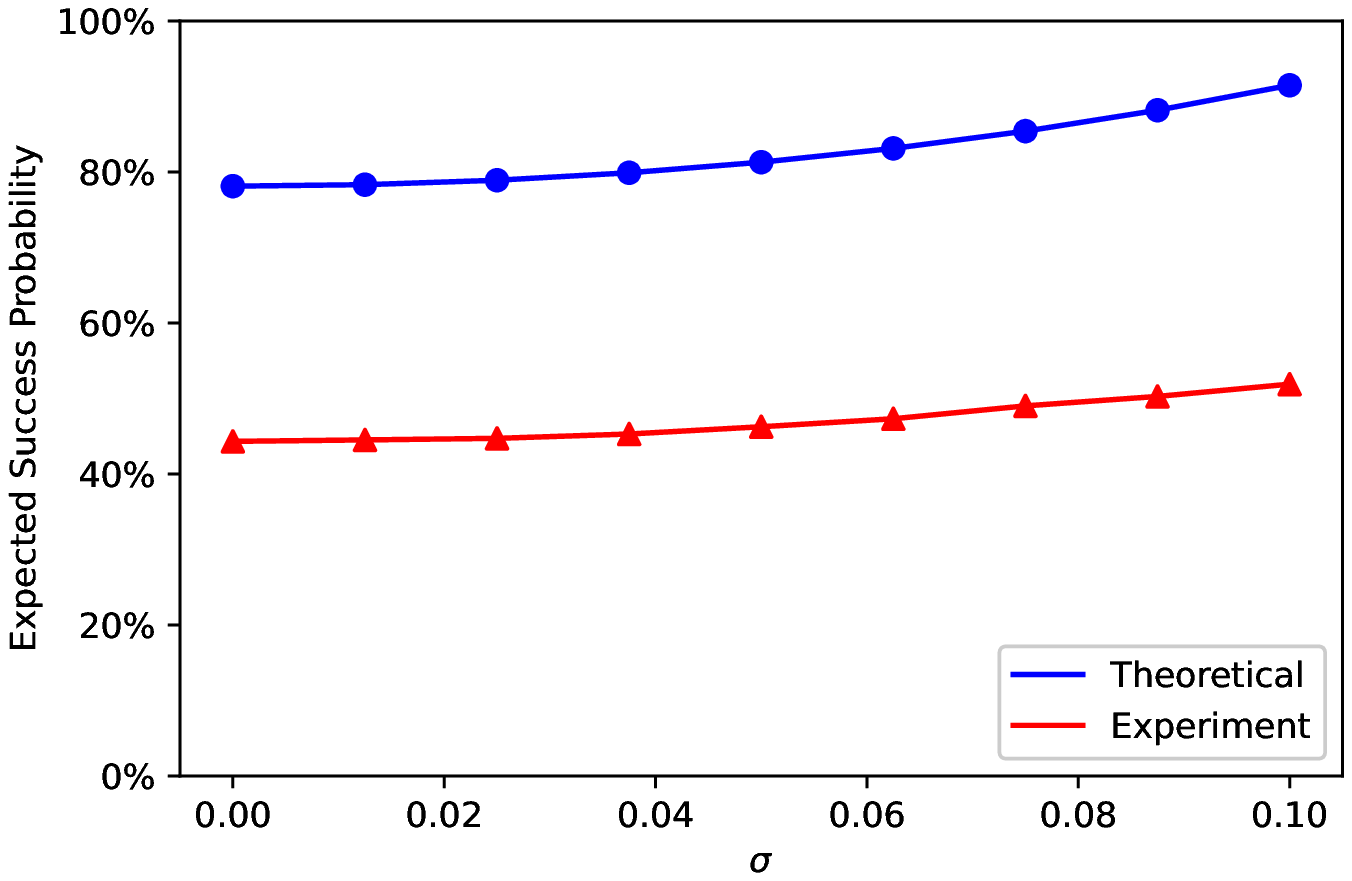}  
	\label{fig:esp}  
}
\subfigure[]{
	\includegraphics[width=0.47\textwidth]{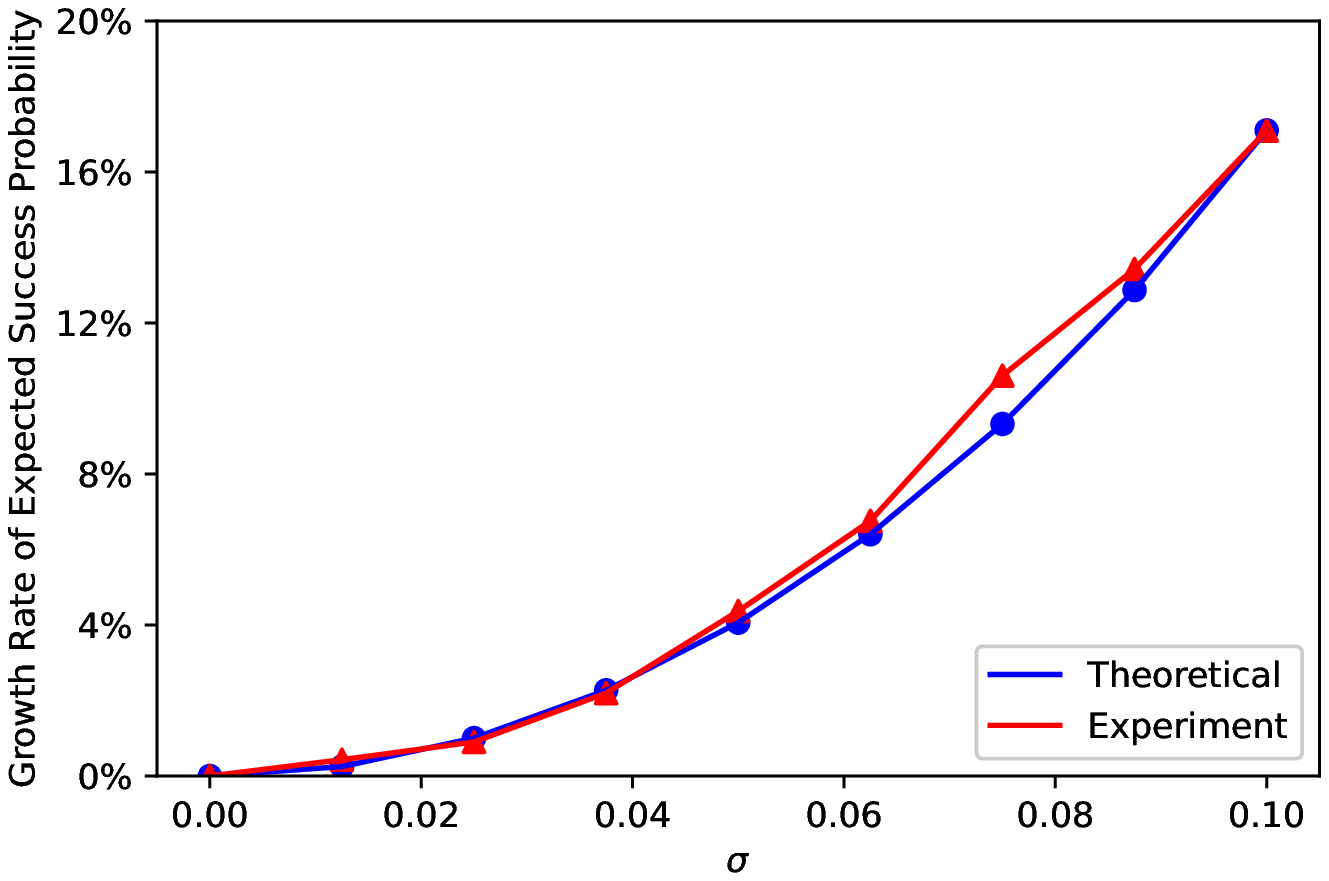}
	\label{fig:desp}  
}
\caption{Quantum search using a single oracle query, for finding one solution among $8$ items on $3$ qubits of IMB's $5$-qubit computing system ibmqx2. (a) Comparison of the theoretical maximum expected success probability and  the experiment result of expected success probability. (b)Comparison of the growth rate of maximum expected success probability than quantum search without prior knowledge and the experimental growth rate of expected success probability than result without prior knowledge. The growth rate here is the ratio of the expected success rate with prior knowledge to the success rate without prior knowledge, then minus $1$.} 
\end{figure*}
\vspace{0.2cm}

\section{Implementation on quantum computing system}
We note that there are many works focusing on implementing Grover's algorithm on quantum computers\cite{figgatt2017complete,brickman2005implementation,shi2017coherence}, they show great techniques decreasing the depth and size of circuit for Grover's algorithm and then obtain great success on the fidelity of experiment circuit. We implement our search algorithm simply since the key of this paper is to make use of prior knowledge, instead of optimizing circuit for search algorithm. We implement our algorithm on a five-qubit quantum computer provided by IBMQ\cite{cross2018ibm}. We consider search among items of size $8$ with one solution. The probability distribution of the solution is a simple biased distribution formed as $(\frac{1}{8}+\sigma,\frac{1}{8}+\sigma,\frac{1}{8}+\sigma,\frac{1}{8}+\sigma,\frac{1}{8}-\sigma,\frac{1}{8}-\sigma,\frac{1}{8}-\sigma,\frac{1}{8}-\sigma)$, for $\sigma\in\{\frac{1}{80},\frac{2}{80},\frac{3}{80},\frac{4}{80},\frac{5}{80},\frac{6}{80},\frac{7}{80},\frac{8}{80}\}$.

The circuit implementation of our algorithm is shown in Figure~\ref{fig:frame}, the corresponding parameter $\theta$ for different distribution is provided in  Appendix.
\begin{figure}[h!] 
  \centerline{
\Qcircuit @C=0.7em @R=0.75em {
   \lstick{\ket{0}}   &   \gate{H}  &\multigate{2}{O_x}&   \gate{H}&   \gate{X}  &\ctrl{2}&   \gate{X}&   \gate{H} \\
   \lstick{\ket{0}}   &   \gate{H}   &\ghost{O_x}&   \gate{H}&   \gate{X}  &\ctrl{1}&   \gate{X}&   \gate{H}  \\
   \lstick{\ket{0}}   &   \gate{R_y(\theta)} &\ghost{O_x}&   \gate{R_y(-\theta)} &   \gate{X} & \gate{Z}&   \gate{X}&   \gate{R_y(\theta)}      \\
}
}
  \caption{A 3-qubit circuit for a 'half-half' distribution in $1$ oracle query. Here $\theta$ depends on the deviation of distribution.}
  \label{fig:frame}
\end{figure}
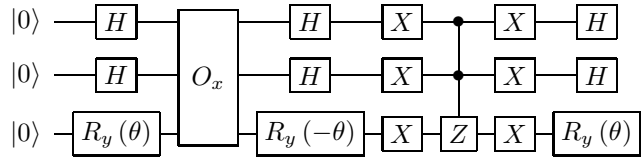

The eight oracles for different solution is implemented directly as shown in Table~\ref{table:oracle}, directly constructed from textbook\cite{nielsen2002quantum}.
The results are output by measurement on standard basis.
For eight biased and one uniform distribution, eight different oracles, we run our algorithm on IBMQ\cite{cross2018ibm} with 'ibmqx2' backend, $8192$ shots for each. We conclude the experiment results in FIG.~\ref{fig:esp} and FIG.~\ref{fig:desp}. It shows that the experimental success probability is lower than theoretical, which is caused by the noise of the device. However, in respect of the growth rate of success probability caused by prior knowledge, the experiment result agrees with the theory well. 

\begin{table}[!htbp]
	\centering 
	\caption{Oracle for different solution} 
	\label{table:oracle}  
	\begin{tabular}{|c|c|c|c|}  
		\hline  
		Solution& Oracle& Solution & Oracle\\
		\hline  
		 & & &\\[-6pt]
		000 &  
\Qcircuit @C=0.7em @R=1em {
  & & &\lstick{\ket{q_0}}   & \gate{X} &  \ctrl{2} &\gate{X} &\qw \\
  & & &\lstick{\ket{q_1}}   & \gate{X} &  \ctrl{1} &\gate{X} &\qw \\
  & & &\lstick{\ket{q_2}}   & \gate{X} &  \gate{Z}    &\gate{X} &\qw      \\
  & & & &\\
}
&100&

\Qcircuit @C=0.7em @R=1em {
  & & &\lstick{\ket{q_0}}   & \gate{X} &  \ctrl{2} &\gate{X} &\qw \\
  & & &\lstick{\ket{q_1}}   & \gate{X} &  \ctrl{1} &\gate{X} &\qw \\
  & & &\lstick{\ket{q_2}}   &\qw &  \gate{Z} &\qw&\qw      \\
  & & & &\\
} \\  		
		\hline  
		 & & &\\[-6pt]
		001 &  
\Qcircuit @C=0.7em @R=1em {
  & & &\lstick{\ket{q_0}}   & \qw &  \ctrl{2} &\qw &\qw \\
  & & &\lstick{\ket{q_1}}   & \gate{X} &  \ctrl{1} &\gate{X} &\qw \\
  & & &\lstick{\ket{q_2}}   & \gate{X} &  \gate{Z}    &\gate{X} &\qw      \\
  & & & &\\
}
&101&

\Qcircuit @C=0.7em @R=1em {
  & & &\lstick{\ket{q_0}}   & \qw &  \ctrl{2} &\qw &\qw \\
  & & &\lstick{\ket{q_1}}   & \gate{X} &  \ctrl{1} &\gate{X} &\qw \\
  & & &\lstick{\ket{q_2}}   &\qw &  \gate{Z} &\qw&\qw      \\
  & & & &\\
} \\  		
		\hline  
		 & & &\\[-6pt]
		010 &  
\Qcircuit @C=0.7em @R=1em {
  & & &\lstick{\ket{q_0}}   & \gate{X} &  \ctrl{2} &\gate{X} &\qw \\
  & & &\lstick{\ket{q_1}}   & \qw &  \ctrl{1} &\qw &\qw \\
  & & &\lstick{\ket{q_2}}   & \gate{X} &  \gate{Z}    &\gate{X} &\qw      \\
  & & & &\\
}
&110&

\Qcircuit @C=0.7em @R=1em {
  & & &\lstick{\ket{q_0}}   & \gate{X} &  \ctrl{2} &\gate{X} &\qw \\
  & & &\lstick{\ket{q_1}}   & \qw &  \ctrl{1} &\qw &\qw \\
  & & &\lstick{\ket{q_2}}   &\qw &  \gate{Z} &\qw&\qw      \\
  & & & &\\
} \\  		
		\hline  
		 & & &\\[-6pt]
		011 &  
\Qcircuit @C=0.7em @R=1em {
  & & &\lstick{\ket{q_0}}   & \qw &  \ctrl{2} &\qw &\qw \\
  & & &\lstick{\ket{q_1}}   & \qw &  \ctrl{1} &\qw &\qw \\
  & & &\lstick{\ket{q_2}}   & \gate{X} &  \gate{Z}    &\gate{X} &\qw      \\
  & & & &\\
}
&111&

\Qcircuit @C=0.7em @R=1em {
  & & &\lstick{\ket{q_0}}   & \qw &  \ctrl{2} &\qw &\qw \\
  & & &\lstick{\ket{q_1}}   & \qw &  \ctrl{1} &\qw &\qw \\
  & & &\lstick{\ket{q_2}}   &\qw &  \gate{Z} &\qw&\qw      \\
  & & & &\\
} \\  		
				\hline
	\end{tabular}
\end{table}

\vspace{0.2cm}
\section{Conclusion}
We consider quantum search with prior knowledge. We provide an optimal quantum search algorithm that achieves the maximum expected success probability with a given number of queries when the probability of solution is given. We show a mathematical proof of the optimality of our algorithm. In particular, the quantum advantage of our algorithm increase as the distribution of solution becomes more biased. We also provide a framework showing how to reduce classical heuristic search algorithm for game tree to our quantum search algorithm.
Since the advantage of our algorithm has been shown by implementation on state-of-the-art device, search problem of larger size and practical search problem, such as searching a game tree, can be implemented hopefully in the future with the development of physical devices.

\bibliographystyle{unsrt}
\bibliography{BibDoi}

\begin{thebibliography}{10}

\bibitem{wiener1994efficient}
Michael~J Wiener.
\newblock {\em Efficient DES key search}.
\newblock School of Computer Science, Carleton Univ., 1994.

\bibitem{daemen1999aes}
Joan Daemen and Vincent Rijmen.
\newblock Aes proposal: Rijndael.
\newblock 1999.

\bibitem{coulom2006efficient}
R{\'e}mi Coulom.
\newblock Efficient selectivity and backup operators in monte-carlo tree
  search.
\newblock In {\em International conference on computers and games}, pages
  72--83. Springer, 2006.

\bibitem{bennett1997strengths}
Charles~H Bennett, Ethan Bernstein, Gilles Brassard, and Umesh Vazirani.
\newblock Strengths and weaknesses of quantum computing.
\newblock {\em SIAM journal on Computing}, 26(5):1510--1523, 1997.

\bibitem{grover1997quantum}
Lov~K Grover.
\newblock Quantum mechanics helps in searching for a needle in a haystack.
\newblock {\em Physical Review Letters}, 79(2):325, 1997.

\bibitem{grover1998quantum}
Lov~K Grover.
\newblock Quantum computers can search rapidly by using almost any
  transformation.
\newblock {\em Physical Review Letters}, 80(19):4329, 1998.

\bibitem{brassard2002quantum}
Gilles Brassard, Peter Hoyer, Michele Mosca, and Alain Tapp.
\newblock Quantum amplitude amplification and estimation.
\newblock {\em Contemporary Mathematics}, 305:53--74, 2002.

\bibitem{zalka1999grover}
Christof Zalka.
\newblock Grover’s quantum searching algorithm is optimal.
\newblock {\em Physical Review A}, 60(4):2746, 1999.

\bibitem{grover2005fixed}
Lov~K Grover.
\newblock Fixed-point quantum search.
\newblock {\em Physical Review Letters}, 95(15):150501, 2005.

\bibitem{mizel2009critically}
Ari Mizel.
\newblock Critically damped quantum search.
\newblock {\em Physical Review Letters}, 102(15):150501, 2009.

\bibitem{yoder2014fixed}
Theodore~J Yoder, Guang~Hao Low, and Isaac~L Chuang.
\newblock Fixed-point quantum search with an optimal number of queries.
\newblock {\em Physical Review Letters}, 113(21):210501, 2014.

\bibitem{silver2017mastering}
David Silver, Julian Schrittwieser, Karen Simonyan, Ioannis Antonoglou, Aja
  Huang, Arthur Guez, Thomas Hubert, Lucas Baker, Matthew Lai, Adrian Bolton,
  et~al.
\newblock Mastering the game of go without human knowledge.
\newblock {\em nature}, 550(7676):354--359, 2017.

\bibitem{silver2018general}
David Silver, Thomas Hubert, Julian Schrittwieser, Ioannis Antonoglou, Matthew
  Lai, Arthur Guez, Marc Lanctot, Laurent Sifre, Dharshan Kumaran, Thore
  Graepel, et~al.
\newblock A general reinforcement learning algorithm that masters chess, shogi,
  and go through self-play.
\newblock {\em Science}, 362(6419):1140--1144, 2018.

\bibitem{montanaro2010quantum}
Ashley Montanaro.
\newblock Quantum search with advice.
\newblock In {\em Conference on Quantum Computation, Communication, and
  Cryptography}, pages 77--93. Springer, 2010.

\bibitem{figgatt2017complete}
Caroline Figgatt, Dmitri Maslov, KA~Landsman, Norbert~Matthias Linke, Shantanu
  Debnath, and C~Monroe.
\newblock Complete 3-qubit grover search on a programmable quantum computer.
\newblock {\em Nature communications}, 8(1):1--9, 2017.

\bibitem{cross2018ibm}
Andrew Cross.
\newblock The ibm q experience and qiskit open-source quantum computing
  software.
\newblock {\em APS}, 2018:L58--003, 2018.

\bibitem{ambainis2017quantum}
Andris Ambainis and Martins Kokainis.
\newblock Quantum algorithm for tree size estimation, with applications to
  backtracking and 2-player games.
\newblock In {\em Proceedings of the 49th Annual ACM SIGACT Symposium on Theory
  of Computing}, pages 989--1002, 2017.

\bibitem{hansen2019faster}
Thomas~Dueholm Hansen, Haim Kaplan, Or~Zamir, and Uri Zwick.
\newblock Faster k-sat algorithms using biased-ppsz.
\newblock In {\em Proceedings of the 51st Annual ACM SIGACT Symposium on Theory
  of Computing}, pages 578--589, 2019.

\bibitem{dantsin2005quantum}
Evgeny Dantsin, Vladik Kreinovich, and Alexander Wolpert.
\newblock On quantum versions of record-breaking algorithms for sat.
\newblock {\em ACM SIGACT News}, 36(4):103--108, 2005.

\bibitem{brickman2005implementation}
K-A Brickman, PC~Haljan, PJ~Lee, M~Acton, L~Deslauriers, and C~Monroe.
\newblock Implementation of grover’s quantum search algorithm in a scalable
  system.
\newblock {\em Physical Review A}, 72(5):050306, 2005.

\bibitem{shi2017coherence}
Hai-Long Shi, Si-Yuan Liu, Xiao-Hui Wang, Wen-Li Yang, Zhan-Ying Yang, and Heng
  Fan.
\newblock Coherence depletion in the grover quantum search algorithm.
\newblock {\em Physical Review A}, 95(3):032307, 2017.

\bibitem{nielsen2002quantum}
Michael~A Nielsen and Isaac Chuang.
\newblock Quantum computation and quantum information, 2002.

\end{thebibliography}

\clearpage
\begin{center}
\vspace*{\baselineskip}
{\textbf{\Large Appendix}}\\
\end{center}
\renewcommand{\theequation}{A\arabic{equation}}
\renewcommand{\thetheorem}{A\arabic{theorem}}
\setcounter{equation}{0}
\setcounter{figure}{0}
\setcounter{table}{0}
\setcounter{section}{0}

\begin{appendix}
\section{How to calculate $\textbf{q}*$}
Here we show how to find the optimal parameters $\textbf{q}^*=(q_1^*,\dots,q_N^*)$  by Lagrange multiplier method when $T=1$.

Note that $\sin^2((2T+1)\arcsin\sqrt{q_i})$ is a polynomial of degree $2T+1$ with respect to $q_i$. When $T=1$, the optimization function is $\sum_{i=1}^N p_i \sin^2(3\arcsin\sqrt{q_i})=\sum_{i=1}^N p_i q_i(3-4q_i)^2$. Note that the when $n\le 4$, the success probability can reach $1$. When $n>4$, Eq~\ref{eq:extra} always reaches its maximum value when $\sum_{i=1}^N q_i=1$ since the function increases monotonously with increase of $q_i$.  The Lagrangian function
\begin{equation}
\mathcal{L} (\bf{q},\lambda)=\sum_{i=1}^N p_i q_i(3-4q_i)^2-\lambda(1-\sum_{i=1}^N q_i)
\end{equation}
By $\nabla\mathcal{L}(\bf{q},\lambda)=0$, we get 
\begin{equation}
\left\{
\begin{aligned}
p_i(48q_i^2-48q_i+9)+\lambda=0,\\
\sum_{i=1}^Nq_i=1.
\end{aligned}
\right.
\end{equation}
which can be simplified to
\begin{equation}\label{eq:lambda}
\sum_{i=1}^N\frac{1}{2}-\sqrt{\frac{1}{16}-\frac{\lambda}{48p_i}}=1
\end{equation}
and 
\begin{equation}
q_i=\frac{1}{2}-\sqrt{\frac{1}{16}-\frac{\lambda}{48p_i}}
\end{equation}
Eq~\ref{eq:lambda} can be solved efficiently by binary search of $\lambda$ since the left side is a monotone function of $x$. 

\section{Proof of the optimality}
\label{app:proof}
\begin{theorem}[Upper bound]\label{thm:upper}
Given distribution $p=(p_1,p_2,\dots,p_N)$ indicating where the solution might be, the expected success probability of quantum algorithm in $T$ oracle queries is upper bounded by
\begin{align}
\max_{\textbf{r}}&\sum_{i=1}^N r_i\sin^2((2T+1)\arcsin\sqrt{r_i})\\
\text{\rm{s.t. }} &\sum_{i=1}^N r_i\le 1,\\
&\forall i:\ 0\le r_i.
\end{align}
\end{theorem}
\begin{proof}
Any quantum algorithm in $T$ queries is formed as a  transformation $U_TO_xU_{T-1}\dots U_1O_xU_0$ applied to $|0\rangle$ and measurement $\{M_i\}$, $i\in\{1,2,\dots,N\}$.

Let $\phi^x_t=U_TO_xU_{T-1}\dots O_xU_{t}U_{t-1}\dots U_0|0\rangle$. 
Let 
\begin{equation}
f(x):
=\left\{
\begin{aligned}
&\sin x,&0\le x\le \pi/2,\\
&1, &\pi/2<x.
\end{aligned}\right.
\end{equation}
For unit vectors $|a\rangle,|b\rangle,|c\rangle$, let $\langle a,c\rangle=2\arcsin(|a-c|)$ denotes the angle between $a$ and $c$, then $\langle a-c|a-c\rangle=4\sin^2(\langle a,c\rangle/2)$. Since $\langle a,c\rangle\le\langle a,b\rangle+\langle b,c\rangle$ by triangle inequality, 
$$\begin{aligned}
&|a-c|\\
=&2\sin(\langle a,c\rangle/2)\\
\le&2\sin(\min(|\langle a,b\rangle/2+\langle b,c\rangle/2,\pi/2))\\
=&2f(|\langle a,b\rangle/2+\langle b,c\rangle/2)
\end{aligned}$$

Quantum states are unit vectors in Hilbert space. Note that $\phi^x_t$ is unit vector for all $t$ and $\phi^x_0=(\phi^x_0-\phi^x_1)+(\phi^x_1-\phi^x_2)+\dots +(\phi^x_{T-1}-\phi^x_T)+\phi^x_T$, $M_x$ is the measurement which outputs $x$. We have
\begin{align*}
&\cos(\langle\phi_0^x,\phi_T^x\rangle)=|\langle\phi_0^x|\phi_T^x\rangle|\\
=&\langle\phi_0^x|M_x M_x|\phi_T^x\rangle+\langle\phi_0^x|(I-M_x)(I-M_x)|\phi_T^x\rangle\\
\le&|M_x|\phi_0^x\rangle| |M_x|\phi_T^x\rangle|+\sqrt{1-|M_x|\phi_0^x\rangle|^2} \sqrt{|M_x|\phi_T^x\rangle|^2}\\
=&\cos(\arcsin(|M_x\phi^x_0|)-\arcsin(|M_x\phi^x_T|))
\end{align*}
Then
\begin{align*}
&\arcsin(|M_x\phi^x_0|)-\arcsin(|M_x\phi^x_T|)\\
\le &\langle\phi^x_0,\phi^x_T\rangle\\
\le &\langle\phi^x_0,\phi^x_1\rangle+\langle\phi^x_2,\phi^x_2\rangle+\dots+\langle\phi^x_{T-1},\phi^x_T\rangle\\
= &2\arcsin(|\phi^x_0-\phi^x_{1}|/2)+\dots+2\arcsin(|\phi^x_{T-1}-\phi^x_{T}|/2).
\end{align*}
Hence the the probability of outputting $x$ with input solution $x$ is 
\begin{align*}
&|M_x \phi_0^x|^2\\
\le&f^2(\arcsin(|M_x\phi_T^x|)+\sum_{t=1}^T 2\arcsin(|\phi^x_{t-1}-\phi^x_{t}|/2)).
\end{align*}

On the other hand
\begin{align*}
&\langle\phi^x_t-\phi^x_{t+1}|\phi^x_t-\phi^x_{t+1}\rangle\\
=&\langle 0|U_0^\dagger\dots U_{t}^\dagger (O_x-I)U_{t+1}^\dagger O_x\dots O_xU_T^\dagger |\\
&U_TO_xU_{T-1}\dots O_xU_{t+1}(O_x-I)U_{t}\dots U_0|0\rangle\\
=&\langle 0|U_0^\dagger\dots U_{t}^\dagger (2|x\rangle\langle x|)(2|x\rangle\langle x|)U_{t}\dots U_0|0\rangle\\
=&4|\langle x|U_{t}\dots U_0|0\rangle|^2.
\end{align*}
Let $u_{t,x}=|\langle x|U_{t}\dots U_0\rangle|^2$, we have
\begin{align*}
&|M_x\phi^x_0|^2\\
\le &f^2(\arcsin(|M_x\phi_T^x|)
+\sum_{t=1}^T 2\arcsin(|\phi^x_{t-1}-\phi^x_{1}|/2))\\
=&f^2( \arcsin\sqrt{q_{T,x}}+2\sum_{t=0}^{T-1}\arcsin\sqrt{u_{0,x}}).
\end{align*}
So the expected success probability is upper bounded by

\begin{align}\label{eq:optimzationq}
\max_{q_{t,x}}&\sum_{x=1}^N p_xf^2( \arcsin\sqrt{u_{T,x}}+2\sum_{t=0}^{T-1}\arcsin\sqrt{u_{0,x}}),\\
\label{eq:le}
\text{s.t. }&\forall t:\sum_{x=1}^N u_{t,x}\le 1,\\
&\forall t,x: u_{t,x}\ge0.
\end{align}

The optimization function Eq~\ref{eq:optimzationq} increases monotonously with $u_{t,x}$, so we can replace the $\le1$ with $=1$ in constraints Eq~\ref{eq:le} without changing the maximum value. By Lemma~\ref{lem:equal}, the function can reach its maximum when $u_{t,x}=u_{0,x}$, which matches Theorem~\ref{thm:upper}.
\end{proof}
\begin{lemma}\label{lem:equal}
The following optimization problem
$$
\begin{aligned}
\max_{\{u_{t,x}\}}&\sum_{x=1}^N p_xf^2(\sum_{t=1}^m\arcsin\sqrt{u_{t,x}}),\\
\text{\rm s.t. }&\forall t:\sum_{x=1}^N u_{t,x}= 1,\\
&\forall t,x: u_{t,x}\ge0.
\end{aligned}
$$
can reach the maximum when $\forall x,t:u_{t,x}=u_{0,x}$.
\end{lemma}
\begin{proof}
Let the function reach its maximum value with the sum of variance
\begin{equation}\label{eq:var}
\sum_{t=1}^m\sum_{x=1}^N\left(u_{t,x}-\sum_{x=1}^N u_{t,x}/n\right)^2
\end{equation} minimized at the same time. If there exists $x,t$ that $u_{t,x}\neq u_{0,x}$, 
assuming $u_{0,1}<u_{1,1}, u_{0,2}>u_{1,2}$ without loss of generality. 

Firstly if $u_{0,1}+u_{1,1}\ge 1$, we can update $u_{0,1}$ and $u_{1,1}$ to $u_{0,1}+\delta$ and $u_{1,1}-\delta$ without decreasing the maximum value when $\delta\le|u_{0,1}-u_{1,1}|/2$, 
since $\arcsin\sqrt{u_{0,1}+\delta}+\arcsin\sqrt{u_{1,1}-\delta}\ge\pi/2$. 

Secondly if $u_{0,1}+u_{1,1}< 1$, we have differentials
\begin{align}
&\mathrm{d}\arcsin\sqrt{u_{0,1}}=1/\sqrt{(1-u_{0,1})u_{0,1}}\mathrm{d } u_{0,1},\\
&\mathrm{d}\arcsin\sqrt{u_{1,1}}=1/\sqrt{(1-u_{1,1})u_{1,1}} \mathrm{d } u_{1,1},
\end{align} and
\begin{equation}
1/\sqrt{(1-u_{0,1})u_{0,1}}-1/\sqrt{(1-u_{1,1})u_{1,1}}>0, 
\end{equation}
when $u_{0,1}+u_{1,1}<1, u_{0,1}<u_{1,1}$.    So when $\delta\le|u_{0,1}-u_{1,1}|/2$
\begin{align*}
\arcsin\sqrt{u_{0,1}+\delta}+\arcsin\sqrt{u_{1,1}-\delta}\\
\ge\arcsin\sqrt{u_{0,1}}+\arcsin\sqrt{u_{1,1}}.
\end{align*}
Similarly, when $\delta\le|u_{0,2}-u_{1,2}|/2$
\begin{align*}
\arcsin\sqrt{u_{0,2}-\delta}+\arcsin\sqrt{u_{1,2}+\delta}\\
\ge\arcsin\sqrt{u_{0,2}}+\arcsin\sqrt{u_{1,2}}.
\end{align*}
So we can choose suitable $\delta$ to update 
$u_{0,1}$ and $u_{1,1}$ or $u_{0,2}$ and $u_{1,2}$ to their average without decreasing the maximum value. 

Then the maximum can be obtained with a smaller summation of variance in Eq.(\ref{eq:var}), which contradicts the assumption!

Hence the maximum value of the optimization function can be obtained when $\forall x,t:u_{t,x}=u_{0,x}$. 
\end{proof}
\section{Optimal parameters $\theta$ for 'half-half' distribution with different deviation $\sigma$}

\begin{table}[!htbp]
	\centering 
	\caption{Parameter $\theta$ for distribution with different deviation $\sigma$}  
	\begin{tabular}{|c|c|c|c|c|}  
		\hline  
		$\sigma$ & $\theta$ && $\sigma$ & $\theta$\\
		\hline
		& & &&\\[-6pt]
		 $\frac{1}{80}$& $1.48725065$ && $\frac{5}{80}$ &$1.12383265$\\[6pt]
		\hline  
		& & &&\\[-6pt]
		 $\frac{2}{80}$& $1.40239865$ && $\frac{6}{80}$ &$1.01471265$\\[6pt]
		\hline  
		& & &&\\[-6pt]
		 $\frac{3}{80}$& $1.31480465$ && $\frac{7}{80}$ &$0.88979265$\\[6pt]
		\hline  
		& & &&\\[-6pt]
		 $\frac{4}{80}$& $1.22272065$ && $\frac{8}{80}$ &$0.73831265$\\[6pt]
		\hline  
	\end{tabular}
\end{table}
\end{appendix}
\end{document}